\documentclass[10pt,journal,compsoc]{IEEEtran}

\usepackage[nocompress]{cite}

\usepackage{amsmath}
\usepackage{amsthm}
\usepackage{graphicx}
\usepackage{booktabs} 
\usepackage{epstopdf}
\usepackage{algorithm}
\usepackage{tabularx}
\usepackage{multirow}
\pdfmapfile{+txfonts.map}
\usepackage{caption}
\usepackage{subcaption}
\usepackage{algpseudocode}
\usepackage{amssymb}
\usepackage{balance}
\newtheorem{prop}{Proposition}

\begin{document}

\title{Dual Metric Learning for Effective and Efficient Cross-Domain Recommendations}

\author{Pan~Li, and~Alexander~Tuzhilin
\IEEEcompsocitemizethanks{\IEEEcompsocthanksitem Pan Li and Alexander Tuzhilin are with the Department of Technology, Operation and Statistics, Stern School of Business, New York University, New York, NY, 10012.\protect\\
E-mail: pli2@stern.nyu.edu}}

\markboth{Journal of \LaTeX\ Class Files,~Vol.~14, No.~8, August~2020}%
{Shell \MakeLowercase{\textit{et al.}}: Bare Demo of IEEEtran.cls for Computer Society Journals}

\IEEEtitleabstractindextext{%
\begin{abstract}
Cross domain recommender systems have been increasingly valuable for helping consumers identify useful items in different applications. However, existing cross-domain models typically require large number of overlap users, which can be difficult to obtain in some applications. In addition, they did not consider the duality structure of cross-domain recommendation tasks, thus failing to take into account bidirectional latent relations between users and items and achieve optimal recommendation performance. To address these issues, in this paper we propose a novel cross-domain recommendation model based on dual learning that transfers information between two related domains in an iterative manner until the learning process stabilizes. We develop a novel latent orthogonal mapping to extract user preferences over multiple domains while preserving relations between users across different latent spaces. Furthermore, we combine the dual learning method with the metric learning approach, which allows us to significantly reduce the required common user overlap across the two domains and leads to even better cross-domain recommendation performance. We test the proposed model on two large-scale industrial datasets and six domain pairs,  demonstrating that it consistently and significantly outperforms all the state-of-the-art baselines.  We also show that the proposed model works well with very few overlap users to obtain satisfying recommendation performance comparable to the state-of-the-art baselines that use many overlap users.
\end{abstract}

}

\maketitle
\IEEEdisplaynontitleabstractindextext
\IEEEpeerreviewmaketitle

\IEEEraisesectionheading{\section{Introduction}\label{sec:introduction}}

\IEEEPARstart{C}{ross} domain recommendation \cite{cantador2015cross,fernandez2012cross} constitutes an important method to tackle sparsity and cold-start problems \cite{schein2002methods,adomavicius2005toward}, thus becoming the key component for online marketplaces to achieve great success in understanding user preferences. Researchers have proposed various mechanisms to provide cross-domain recommendations through factorization approaches \cite{singh2008relational,li2009can,hu2013personalized,loni2014cross,fernandez2019addressing,wang2019solving} and transfer learning methods \cite{pan2010survey,hu2018conet,lian2017cccfnet,li2020ddtcdr,gao2019cross} that \textit{transfer} user preferences from the source domain to the target domain \cite{pan2010survey}. For instance, if a user watches a movie, we may recommend the novel on which that movie is based to that user \cite{li2009can}. 

However, most of the existing methods only focus on the unidirectional learning process for providing cross-domain recommendations, i.e., utilizing information from the source domain to improve the recommendations in the target domain. This is unfortunate since the duality nature of the cross-domain recommendation task is not explored, and it would also be beneficial to leverage user preference information from the target domain to predict user behaviors in the source domain. For example, once we know the type of books that the user would like to read, we can recommend movies on related topics to form a loop for better recommendations in both domains. 

Dual learning models \cite{qin2020dual}, meanwhile, have been shown in the previous research to be effective and efficient in multiple applications to address duality nature in the learning tasks, including domain adaption \cite{zhong2009cross}, transfer learning \cite{long2012dual,wang2011cross} and machine translation \cite{he2016dual}. Therefore in this paper, we proposed to apply the dual learning mechanism to the problem of cross-domain recommendations, with the assumption that improving the learning process in one domain would also help the learning process in the other domain. We address this assumption by proposing a unifying mechanism that extracts the essence of preference information in each domain and bidirectionally transfer user preferences between different domains to improve the recommendation performance across different domains simultaneously.

Note that existing cross-domain recommendation models typically require a large amount of overlap users between different domains as 'pivots', in order to learn the relations of user preferences and produce satisfying recommendation performance \cite{fernandez2012cross,cremonesi2014cross}. These overlap users have consumed items in both categories, such as watching a movie and reading a book. However, collecting sufficient amount of overlap users could be difficult and expensive to accomplish in many applications. For example, there might be only a limited number of users who purchased books and digital music altogether on Amazon. Therefore, it is important to overcome this problem and minimize the required common user overlap across the two domains in cross-domain recommendations.

To address this issue, one potential solution is to construct those 'pivots' for cross-domain recommendations based on not only the overlap users, but also those users with similar preferences, with the assumption that if two users have similar preferences in a certain domain, their preferences would also be similar across other domains. While the information of user pairs with similar preferences are not explicitly recorded in the dataset, metric learning models \cite{yang2006distance} are capable of learning the distance metric between users and items in the \textit{latent} space through shared representations, and capturing user pairs with similar preferences accordingly \cite{hsieh2017collaborative}. In particular, the metric learning model would pull latent embeddings of overlap user pairs close to each other in the latent space, and also pull those similar user pairs close to each other due to triangle inequality of the distance measure \cite{kulis2012metric}. Besides that, as the metric learning model is capable of learning a joint latent metric space to encode not only users’ preferences but also the user-user and item-item similarity information, it also significantly contributes to the improvements of recommendation performance \cite{hsieh2017collaborative}.

Therefore, in this paper we also propose to supplement the aforementioned dual learning mechanism with the metric learning approach to effectively identify the latent relations in user preferences across different domains using the minimal amount of overlapping users. Moreover, we empirically demonstrate that, by iteratively updating the dual metric learning model, we simultaneously improve recommendation performance over both domains and outperform all the state-of-the-art baseline models. And crucially, the proposed Dual Metric Learning (DML) model provides these strong recommendation performance results with only few overlapping uses.

In this paper, we make the following contributions. We
\begin{itemize}
\item propose to apply the dual learning mechanism for providing cross-domain recommendations to address the duality of learning tasks
\item propose to incorporate the metric learning model into the dual learning mechanism to reduce the requirement of having large amounts of overlapping users
\item present a novel cross-domain recommendation model DML and empirically demonstrate that it significantly and consistently outperforms the state-of-the-art approaches across all the experimental settings
\item provide theoretical foundations to (a) explain why DML reduces the need for the user overlap information; (b) demonstrate the convergence condition for the simplified case of our model; and (c) illustrate that the proposed model can be easily extended to multiple-domain recommendation applications
\item conduct extensive ablation experiments to study (a) the cross-domain recommendation performance of different approaches; (b) the impact of different architectures of neural networks; (c) the influence of including and excluding explicit feature information during the representation learning stage; (d) the impact of different numbers of overlap users included in the model; (e) the sensitivity of hyperparameters, particularly the dimension of latent embeddings; (f) the scalability of the proposed approach; and (g) the convergence behavior of the proposed approach.
\end{itemize}

\section{Related Work}
The work related to the proposed model comes from the following research areas that we review in this section: cross domain recommendations, dual learning, metric-learning for recommendations and deep-learning based recommendations. 

\subsection{Cross Domain Recommendations}
Cross-domain recommendation approach \cite{fernandez2012cross} constitutes a powerful tool to deal with the data sparsity problem. Typical cross domain recommendation models are extended from single-domain recommendation models, including CMF \cite{singh2008relational}, CDCF \cite{li2009can,hu2013personalized}, CDFM \cite{loni2014cross}, HISF \cite{do2019unveiling}, NATR \cite{gao2019cross}, DTCDR \cite{zhu2019dtcdr}, Canonical Correlation Analysis \cite{sahebi2017cross}, Dual Regularization \cite{wu2018dual}, Preference Propagation \cite{zhao2019cross} and Deep Domain Adaption \cite{kanagawa2019cross}.  These approaches assume that different patterns characterize the way that users interact with items of a certain domain and allow interaction information from one domain to inform recommendation in the other domain. 

The idea of information fusion also motivates the use of transfer learning methods \cite{hu2018conet,lian2017cccfnet,gao2019cross,fu2019deeply,hu2019transfer} that transfer extracted information from the source domain to the target domain. The existing transfer learning methods for cross-domain recommendation include Collaborative DualPLSA \cite{zhuang2010collaborative}, JDA \cite{long2013transfer} and RB-JTF \cite{song2017based}.

However, these models do not fundamentally address the relationship between different domains, for they do not improve recommendation performance of both domains simultaneously, thus might not release the full potential of utilizing the cross-domain user interaction information. Also, they do not explicitly model user and item features during recommendation process, and usually require large amount of overlap users. In this paper, we propose a novel dual  metric learning mechanism combining with autoencoders to overcome these issues and significantly improve recommendation performance.

\subsection{Dual Learning}
Transfer learning \cite{pan2010survey} deals with the situation where the data obtained from different resources are distributed differently. It assumes the existence of common knowledge structure that defines the domain relatedness, and incorporate this structure in the learning process by discovering a shared latent feature space in which the data distributions across domains are close to each other. The existing transfer learning methods for cross-domain recommendation include Collaborative DualPLSA \cite{zhuang2010collaborative}, Joint Subspace Nonnegative Matrix Factorization \cite{liu2013multi}  and JDA \cite{long2013transfer} that learn the latent factors and associations spanning a shared subspace where the marginal distributions across domains are close. In recent years, researchers have utilized up-to-date deep learning techniques to design the transfer learning models for providing cross-domain recommendations, including CATN \cite{zhao2020catn} and MMT-Net \cite{krishnan2020transfer}, and they have all achieved satisfying recommendation performance.

In addition, to exploit the duality between these two distributions and to enhance the transfer capability, researchers propose the dual transfer learning mechanism \cite{long2012dual,zhong2009cross,wang2011cross} that simultaneously learns the marginal and conditional distributions. Recently, researchers manage to achieve great performance on machine translation with dual-learning mechanism \cite{he2016dual,xia2017dual,wang2018dual}. All these successful applications address the importance of exploiting the duality for mutual reinforcement. However, none of them apply the dual transfer learning mechanism into cross-domain recommendation problems, where the duality lies in the symmetrical correlation between source domain and target domain user preferences. In this paper, we utilize a novel dual-learning mechanism and significantly improve recommendation performance.

\subsection{Metric Learning for Recommendations}
Metric learning \cite{yang2006distance} learns a distance metric that preserves the distance relation among the training data and assigns shorter distances to semantically similar pairs. Classical metric learning approach attempts to learn a Mahalanobis distance metric
\begin{equation}
d_{A}(x_{i}, x_{j}) = \sqrt{(x_{i}-x_{j})^{T}A(x_{i}-x_{j})}
\end{equation}
where $A\in R^{m\times m}$ is a positive semi-definite matrix \cite{xing2003distance} that projects each input $x$ to the Euclidean space $R^{m}$. The above global optimization pulls all similar pairs together, and pushes dissimilar pairs apart. 

There are many different ways to learn the distance metric for providing recommendations, including convex optimization, kernel methods and neural networks \cite{kedem2012non,wang2011metric,khoshneshin2010collaborative,chen2012playlist,feng2015personalized}. Recently, researchers proposed Collaborative Metric Learning (CML) model \cite{hsieh2017collaborative} that projects users and items into the shared latent space. Therefore, it is capable of modeling the intensity and the heterogeneity of the user–item relationships in the implicit feedback. Researchers in \cite{tay2018latent,park2018collaborative,vinh2020hyperml} further improve the framework of CML and obtain better recommendation performance.

However, none of the existing metric learning works focus on the application of cross-domain recommendations to reduce the requirement of overlap user information. In this paper, we use metric learning to map the users into the shared latent space across different domains. By effectively utilizing information of overlap users, we achieve satisfying cross-domain recommendation performance using only very few users, thus significantly improving the practicality of the cross-domain recommendation models.

\subsection{Deep Learning-based Recommendations}
Recently, deep learning has been revolutionizing the recommendation architectures dramatically and brings more opportunities to improve the performance of existing recommender systems. To capture the latent relationship between users and items, researchers propose to use deep learning based recommender systems \cite{zhang2019deep,wang2015collaborative}, especially embedding and autoencoding methods \cite{he2017neural,sedhain2015autorec,wu2016collaborative,li2017collaborative,liang2018variational,li2019latent,9222083,li2020purs} to extract the latent essence of user-item interactions for the better understanding of user preferences.

However, user preferences in different domains are learned separately without exploiting the duality for mutual reinforcement, for researchers have not addressed the combination of deep learning methods with dual learning mechanism in recommender systems, which learns user preferences from different domains simultaneously and further improves the recommendation performance. To this end, we propose a deep dual metric learning model that captures latent interactions across different domains. The effectiveness of dual learning over the existing methods is empirically demonstrated by extensive experiments.

\section{Model}
In this section, we introduce the proposed Dual Metric Learning (DML) model for improving cross-domain recommendation performance and reducing the need of overlap information. We illustrate the proposed model in Figure \ref{model} and present the framework in Algorithm \ref{dml}. We will now describe the details of DML model in the following sections.

\begin{figure*}
\centering
\includegraphics[width=0.75\textwidth]{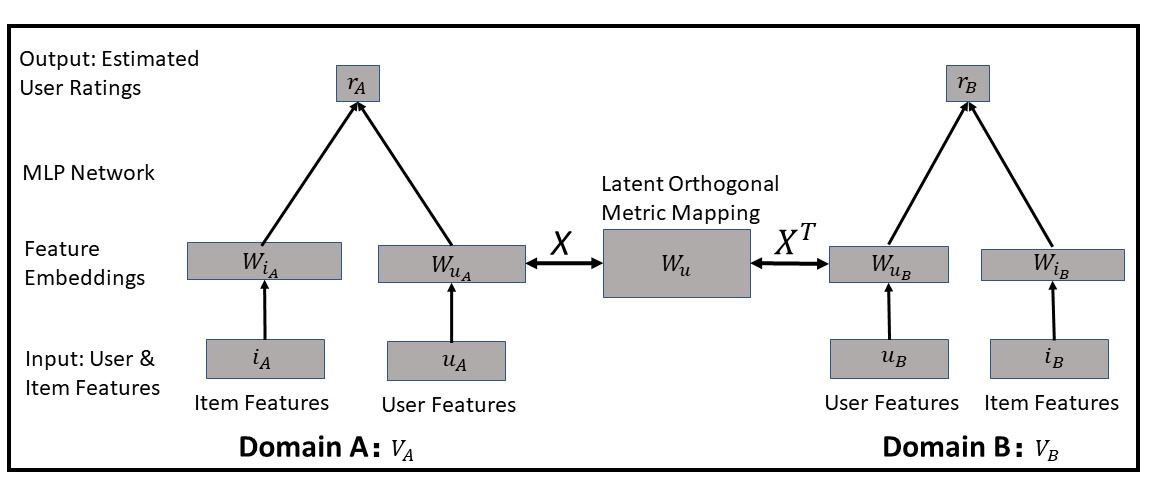}
\caption{Dual Metric Learning for Cross-Domain Recommendations}
\label{model}
\end{figure*}

\begin{algorithm}[!]
\caption{Dual Metric Learning (DML) Model}
\begin{algorithmic}[1]
\State \textbf{Input}: Domain $A$ and $B$, user embeddings $W_{u_{A}}$ and $W_{u_{B}}$, item embeddings $W_{i_{A}}$ and $W_{i_{B}}$, ratings $r_{A}$ and $r_{B}$, neural recommendation models $RS_{A}$ and $RS_{B}$, orthogonal mapping $X$, overlap users $ou_{A}=ou_{B}$
\Repeat
\For{($W_{u_{A}}$, $W_{i_{A}}$, $r_{A}$) in $A$, ($W_{u_{B}}$, $W_{i_{B}}$, $r_{B}$) in $B$}
\State Estimate $r_A'= RS_{A}(W_{u_{A}},W_{i_{A}})$ and $r_B'= RS_{B}(W_{u_{B}},W_{i_{B}})$
\State Backpropogate $L_{A} = ||r_A - r_A'||$ and $L_{B} = ||r_B - r_B'||$
\State Update $RS_{A}$ and $RS_{B}$
\EndFor
\For{$W_{ou_{A}}$, $W_{ou_{B}}$ in Overlap Users}
\State Backpropogate $L_{o_{A}} = ||XW_{ou_{A}} - W_{ou_{B}}||$ and $L_{o_{B}} = ||X^{T}W_{ou_{B}} - W_{ou_{A}}||$
\State Update and Orthogonalize $X$
\EndFor
\For{($W_{u_{A}}$, $W_{i_{A}}$, $r_{A}$) in $A$, ($W_{u_{B}}$, $W_{i_{B}}$, $r_{B}$) in $B$}
\State Estimate $r^{*}_{A}= RS_{B}(XW_{u_{A}},W_{i_{A}})$ and $r^{*}_{B}= RS_{A}(X^{T}W_{u_{B}},W_{i_{B}})$
\State Backpropogate $L_{A^{*}} = ||r_A - r^{*}_{A}||$ and $L_{B^{*}} = ||r_B - r^{*}_{B}||$
\State Update $RS_{A}$ and $RS_{B}$
\EndFor
\Until{convergence}
\end{algorithmic}
\label{dml}
\end{algorithm}

\subsection{User and Item Embeddings}
To effectively extract latent user preferences and efficiently model user and item features, we utilize the autoencoder framework to learn the latent representations of users and items, which transforms the heterogeneous and discrete feature vectors into continuous feature embeddings. The goal is to train two separate neural networks: the encoder network that maps feature vectors into the latent space, and the decoder network that reconstructs feature vectors from latent embeddings. Due to effectiveness and efficiency of the training process, we formulate both the encoder and the decoder as multi-layer perceptron (MLP).

Specifically, we denote features of user $a$ as $u_{a} = \{u_{a_{1}},u_{a_{2}},\cdots,u_{a_{m}}\}$ and features of item $b$ as $i_{b} = \{i_{b_{1}},i_{b_{2}},\cdots,i_{b_{n}}\}$, where $m$ and $n$ are the respective dimensions of feature vectors. MLP learns the hidden representations by optimization reconstruction loss $L$:
\begin{equation}
L = ||u_{a}-MLP_{dec}(MLP_{enc}(u_{a}))||
\end{equation}
where $MLP_{enc}$ and $MLP_{dec}$ represents the MLP network for the encoder and the decoder. In this step, we train the autoencoder separately for users and items in different domains to avoid the information leak between domains. We subsequently obtain the user embeddings and item embeddings as $W_{u_{a}}=MLP_{enc}(u_{a})$ and $W_{i_{b}}=MLP_{enc}(i_{b})$ for the recommendation stage.

Note that regularization technique is typically crucial to the feasibility and efficient optimization of the metric learning model. Especially when the data points are distributed in a high-dimensional space, they might spread too widely because of the curse of dimensionality and thus become ineffective to be mapped into the same shared latent space. Therefore, we restrict all the user and item embeddings $||W_{u}||$ and $||W{i}||$ within the unit ball, i.e., $||W_{u}||^{2}\leq 1$ and $||W{i}||^{2}\leq 1$, to ensure the robustness of the learned metric. As such, we do not need to impose additional regularization during the subsequent metric learning process.

\subsection{Metric Learning}
In this section, we introduce the proposed DML recommendation model to efficiently capture relative preferences between different domains through overlap users. To start with, we identify and extract overlap users and their embeddings in two domains as a set of embedding-embedding pairs $S$ that match with each other. Note that in the previous stage, we separately train user embeddings and item embeddings for each domain. Due to the heterogeneity of user preferences in different domains, the distribution of user embeddings in each domain should also be different. The underlying assumption we make is that, if two users have similar interest in one domain, then those two users would have similar interest in the other domain as well. We aim at addressing this assumption to incorporate consumption behaviors in one domain for better predictions of the other domain. Therefore, these overlap users can be used as `pivots' to learn the relations of user preferences and behaviors in different domains.

We aim to learn a transitional metric learning mapping to encode these relationships in a shared latent space, and the learned mapping should pull the overlap user embedding pairs closer and push the other pairs relatively further apart. According to the triangle inequality, this learning process will also cluster the items that are purchased by the same user across different domains together, and the users who purchase similar items across different domains together. Eventually, the nearest neighbor points for any given user in the shared latent space will become its representation in the other domains respectively. Through the metric learning mapping that effectively captures the relations between overlap users, we propagate these relations not only to overlap user-user pairs, but also to those user-item pairs and non-overlap user-user pairs for which we did not directly observe such relationships.

Distinguished from the classical metric learning settings, the proposed DML model restricts this transitional mapping to be an orthogonal mapping. We point out that orthogonality is important for cross-domain recommendations for the following reasons:
\begin{itemize}
\item Orthogonality preserves the similarities between embeddings of different users during the metric learning process since orthogonal transformation preserves inner product of latent vectors. 
\item Orthogonality mapping $X$ will automatically give the inverse mapping as $X^T$, because $Y = X^T(XY)$ holds for any given orthogonal mapping $X$, thus making inverse mapping matrix equivalent to its transpose. This simplifies the learning procedure and enables the dual learning process. 
\item Orthogonality restricts the parameter space of the transitional mapping and thus requires less overlap user information to determine this mapping function, as we will discuss in the next section.
\end{itemize}

Most importantly, the proposed DML model utilizes dual learning mechanism to update the recommender system for both domains simultaneously. We use the metric learning output of user embeddings in the source domain as the input of recommender system in the target domain, while in turn using the output of the target domain as the input of the source domain. In this way, we could improve the metric learning mapping as well as recommender systems for both domains simultaneously in each iteration. The learning process can then be repeated iteratively to obtain better metric mappings and recommender systems each time until the convergence criterion is met. As a result, it pushes the metric learning framework to better capture user preferences and thus provide even better recommendation performance. We illustrate the dual orthogonal metric learning process in Figure \ref{function}.

\begin{figure}
\centering
\includegraphics[width=0.35\textwidth]{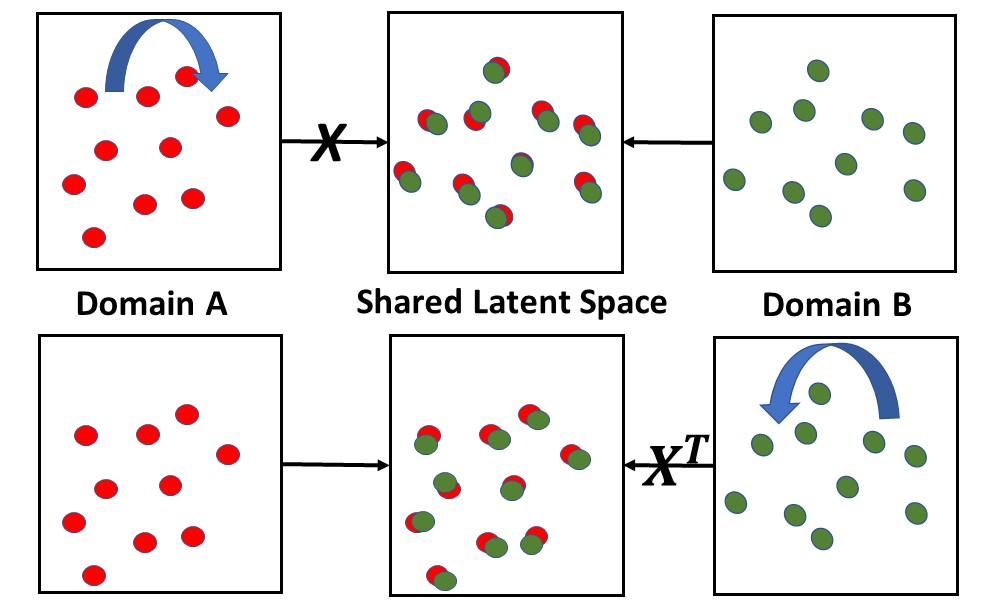}
\caption{Dual Metric Learning} 
\label{function}
\end{figure}

Specifically, we denote the user embeddings for domain $A$ and $B$ as $W_{u_{A}}$, $W_{u_{B}}$ that we obtained in the previous section. We also denote the overlap users between two domains as $ou_{A}=ou_{B}$, and their user embeddings $W_{ou_{A}}$, $W_{ou_{B}}$ correspondingly. The goal of proposed DML model is to find the optimal mapping matrix $X$ to minimize the sum of squared Euclidean distances between the mapped user embeddings $XW_{ou_{A}}$ and the target user embeddings $W_{ou_{B}}$ for the same overlap users: 
\begin{equation}
L_{X} = argmin_{X} \sum_{W_{ou_{A}},W_{ou_{B}}\in ou_{A},ou_{B}} ||XW_{ou_{A}}-W_{ou_{B}}||^{2}
\end{equation}
This optimization is equivalent to its dual form:
\begin{equation}
L_{X^{T}} = argmin_{X} \sum_{W_{ou_{A}},W_{ou_{B}}\in ou_{A},ou_{B}} ||W_{ou_{A}}-X^{T}W_{ou_{B}}||^{2}
\end{equation}
We constrain the metric learning mapping $X$ to be an orthogonal mapping (i.e. $XX^{T}=X^{T}X=I$), which serves to enforce structural invariance of user preferences in each domain, while preventing a degradation in mono-domain recommendation performance for learning better transitional mappings. We optimize Equation (3) and (4) simultaneously to learn the orthogonal metric mapping.

\subsection{Dual Learning}
To further improve performance of recommendations, we propose to combine dual learning mechanism with cross domain recommendations, where we transfer user preferences between the source domain and the target domain simultaneously. Consider two different domains $D_{A}$ and $D_{B}$ that contain user-item interactions as well as user and item features. To obtain better understanding of user preferences in $D_{A}$, we also utilize external user preference information from $D_{B}$ and combine them together. Similarly, we could get better recommendation performance in $D_{B}$ if we utilize user preference information from $D_{A}$ at the same step. To leverage duality of the two transfer learning based recommendation models and to improve effectiveness of both tasks simultaneously, we conduct dual learning recommendations for the two models \textit{together} and update the recommendation models for both domains accordingly.  

Specifically, we propose to model cross-domain user preferences utilizing the metric mapping $X$ learned in the previous stage and taking into account heterogeneous user behavior in both domains. In that sense, we estimate user ratings in domain pairs $(A,B)$ as follows:
\begin{equation}
r^{*}_A= RS_{A}(XW_{u_{A}},W_{i_{A}})
\end{equation}
\begin{equation}
r^{*}_B= RS_{B}(X^{T}W_{u_{B}},W_{i_{B}})
\end{equation}
where $W_{u_{A}},W_{i_{A}},W_{u_{B}},W_{i_{B}}$ represents user and item embeddings and $RS_{A}, RS_{B}$ stand for the domain-specific neural recommendation model for domain A and B respectively.  Here, we intentionally transfer the user representations from one domain to the other and force the recommendation to learn the user preferences from the other domain in order to provide better recommendations. As shown in Figure \ref{model}, dual learning entails the transfer loop across two domains and the learning process goes through the loop iteratively. It is also important to study the convergence property of the model, which we discuss in the next section.

\subsection{Summary of DML Model}
To conclude, we propose the DML method to model the overlap user information for providing better cross-domain recommendations. Specifically, we impose the orthogonal constraint to the metric learning mapping and perform dual learning mechanism to effectively and efficiently learn the mapping function as well as recommender systems at the same time. In the final stage, we concatenate the item embeddings with the user embeddings in the shared latent space for subsequent rating estimation. We feed them into the Multi-Layer Perceptron (MLP) network with dropout to provide final cross-domain recommendations due to its superior representational capacity and ease of training \cite{he2017neural}. 

Compared to the previous approaches, the proposed DML model has the following advantages. It
\begin{itemize}
\item utilizes dual learning mechanism to enable bidirectional transfer of user preferences that improves recommendation performance simultaneously in both domains over time;
\item combines dual learning with metric learning to effectively obtain latent relations across different domains, thus significantly reducing the amount of overlap users required to achieve strong performance;
\item transfers latent representations of features and user preferences, instead of explicit information, between different domains to capture latent and complex interactions between the users and items;
\item learns the latent \textit{orthogonal} mapping function across the two domains, which is capable of preserving similarities of user preferences and computing the inverse mapping function efficiently.
\end{itemize}

We will show in the following section that the proposed model is capable of providing state-of-the-art cross-domain recommendation performance and significantly reducing the need for user overlap information.

\section{Theoretical Analysis}
In this section, we provide the theoretical foundations to justify the design of our DML model. Specifically, we point out that DML reduces the need for overlap information, guarantees to converge under certain conditions and could be easily extended to multi-domain scenarios.

\subsection{Reducing the need for overlap information}
By framing the cross-domain recommendation problem as the metric learning problem, we do not need to extract explicit user behaviors from both domains; instead, we only need to learn the metric learning mapping that represents the relations between the two domains. In addition, by mapping both user and item features into the latent space, we overcome the problem of discretization in the feature space and mathematically formulate the problem as follows: given a set of $k$-dimensional feature vectors $\{W_{u_{A_{1}}}, W_{u_{A_{2}}}, \cdots, W_{u_{A_{n}}}\}$ in domain A and $\{W_{u_{B_{1}}}, W_{u_{B_{2}}}, \cdots, W_{u_{B_{n}}\}}$ in domain B, the goal is to identify the matrix $X$ that $XW_{u_{A_{i}}}=W_{u_{B_{i}}}$ holds for every $i$ from $1$ to $n$. In addition, we restricts this matrix $X$ to be orthogonal, i.e., $XX^{T}=I$. From basic linear algebra \cite{meyer2000matrix}, we know that the unique solution exists when the number of restricted equations $n=\frac{k(k-1)}{2}$. Therefore, when the dimension of feature embeddings is $k$, we only need $\frac{k(k-1)}{2}$ number of overlap users to determine the metric learning mapping. As shown in the results section, we conduct sensitivity analysis to study the impact of different embedding dimensions towards cross domain recommendation performance, which shows that we can achieve reasonable and satisfying performance using 16 dimensions of latent vectors. Therefore, our theoretical analysis shows that we could determine the metric learning mapping using only 120 overlap users for our recommendation tasks, and experiment results further confirm that we could provide satisfying recommendation performance using as little as 8 overlap users.

\subsection{Convergence Analysis}
In this section, we present the convergence theorem of classical recommender systems with dual learning mechanism. We denote the rating matrices as $V_{A},V_{B}$ for the two domains A and B respectively. The goal is to find the approximate non-negative factorization matrices that simultaneously minimizes the reconstruction loss. We conduct multiplicative update algorithms \cite{lin2007convergence} for the two reconstruction loss iteratively given $\alpha$, $X$, $V_{A}$ and $V_{B}$.
\begin{equation}
\begin{split}
\displaystyle \min_{W_{A},W_{B},H_{A},H_{B}} & |V_{A} - (1-\alpha) W_{A}H_{A} - \alpha XW_{B}H_{B}| \\
  + & |V_{B} - (1-\alpha) W_{B}H_{B} - \alpha X^TW_{A}H_{A}|
\end{split}
\end{equation}
\begin{prop}
Convergence of the iterative optimization of dual matrix factorization for rating matrices $V_{A}$ and $V_{B}$ in (7) is guaranteed.
\end{prop}
\begin{proof}
Combining the two parts of objective functions and eliminating $W_{A}H_{A}$ term, we will get
\begin{equation}
(1-\alpha) V_{B} - \alpha X^TV_{A} = (1- 2\alpha)W_{B}H_{B}
\end{equation}
Similarly, when we eliminate $W_{B}H_{B}$ term, we will get
\begin{equation}
(1-\alpha) V_{A} - \alpha XV_{B} = (1- 2\alpha)W_{A}H_{A}
\end{equation}
Based on the analysis in \cite{lee2001algorithms} for the classical single-domain matrix factorization, to show that the repeated iteration of update rules is guaranteed to converge to a locally optimal solution, it is sufficient to show that 
$$ \left\{
\begin{array}{rcl}
2\alpha - 1 < 0       &      &  (a)\\
(1-\alpha) V_{B} - \alpha X^TV_{A} \ge \textbf{0}    &      & (b)\\
(1-\alpha) V_{A} - \alpha XV_{B} \ge \textbf{0}     &      & (c)\\
\end{array} \right. $$
where \textbf{0} stands for the zero matrix. Condition (a) is an intuitive condition which indicates that the information from the target domain dominantly determines the user preference, while information transferred from the source domain only serves as the regularizer in the learning period. To fulfill the seemingly complicated condition (b) and (c), we recall that $V_{A}$ and $V_{B}$ are ratings matrices, which are non-negative and bounded by the rating scale $k$. We design two ''positive perturbation" $V_{A}' = V_{A} + mk\textbf{1}$ and $V_{B}' = V_{B} + mk\textbf{1}$ where $m$ is the rank of the mapping matrix $X$. Condition (a) is independent of the specific rating matrices, and we could check that $V_{A}'$ and $V_{B}'$ satisfies condition (b) and (c). Thus, the matrix factorization of $V_{A}'$ and $V_{B}'$ is guaranteed convergence; to reconstruct the original matrix $V_{A}$ and $V_{B}$, we only need to deduce the ''positive perturbation'' term $m$, so the original problem is also guaranteed convergence.
\end{proof}

In this paper, to capture latent interactions between users and items, we use the neural network based recommendation approach, therefore it still remains unclear if this convergence happens to the proposed model. However, our hypothesis is that even in our case we will experience similar convergence process guaranteed in the classical matrix factorization case, as stated in the proposition. We test and validate this hypothesis in the result section.

\subsection{Extension to Multiple Domains}
In previous sections, we describe the proposed model with special focus on the idea of combining dual metric learning mechanism with latent embedding approaches. We point out that the proposed model not only works for cross domain recommendations between two domains, but can easily extend to recommendations between multiple domains as well. Consider $N$ different domains $D_{1}, D_{2}, D_{3}, \cdots, D_{N}$. To provide cross-domain recommendations for these domains, we only need to obtain the $N-1$ orthogonal mapping learning mapping $X_{12}, X_{23}, \cdots, X_{(N-1)N}$, and the latent orthogonal transfer matrix $X_{jk}$ between domain $D_{j}$ and $D_{k}$ could be obtained as $X_{jk}=X_{j(j+1)}*X_{(j+1)(j+2)}*\cdots *X_{(k-1)k}$. Therefore, the proposed model is capable of providing recommendations for multiple-domain application effectively and efficiently.

\section{Experiment}
To demonstrate the superiority of our proposed method, we conducted extensive experiments on two real-world datasets in this paper to study the following research questions: (1) the cross-domain recommendation performance of different approaches; (2) the impact of different architectures of neural networks; (3) the influence of including and excluding explicit feature information during the representation learning stage; (4) the impact of different numbers of overlap users included in the model; (5) the sensitivity of hyperparameters, particularly the dimension of latent embeddings; (6) the scalability of the proposed approach; and (7) the convergence behavior of the proposed approach. Our models and codes have been made publicly available \footnote{https://github.com/lpworld/DOML}.
 
\subsection{Dataset}
We evaluate the cross-domain recommendation performance of the proposed on two sets of large-scale industrial datasets: the Imhonet dataset \cite{DBLP:conf/cla/BobrikovNI16}, which is obtained directly from the user logs of the European online recommendation service Imhonet, containing three domains of user and item features as well as rating information, namely Book, Movie and Music; and the Amazon dataset \cite{ni2019justifying} \footnote{https://nijianmo.github.io/amazon/index.html} which consists of user purchase actions and rating information collected from the Amazon platform, and we select three domains with sufficient amount of overlap users for conducting our experiments: Beauty, Fashion and Art Crafts. Overlap users across different domains can be identified by the same user ID. We normalize the scale of ratings to between 0 and 1 for both datasets. The basic statistics about these three datasets are shown in Table \ref{imhonet} and Table \ref{amazon}. There are around 2,000 overlap users between each domain pair in the Imhonet dataset, and around 30,000 overlap users between each domain pair in the Amazon dataset.

\begin{table}
\centering
\begin{tabular}{|llll|}
\hline
Datasets & Book & Movie & Music \\ \hline
\# Users & 804,285 & 959,502 & 45,962 \\ 
\# Items & 182,653 & 79,866 & 183,114 \\ 
\# Records & 223,007,805 & 51,269,130 & 2,536,273 \\ 
Sparsity & 0.0157\% & 0.0669\% & 0.0301\% \\ \hline
\end{tabular}
\caption{Descriptive Statistics for the Imhonet Dataset}
\label{imhonet}
\end{table}

\begin{table}
\centering
\begin{tabular}{|l l l l|}
\hline
Domain & Beauty & Fashion & Art Crafts \\ \hline
\#Users & 324,038 & 749,233 & 1,579,230 \\ 
\#Items & 32,586 & 186,189 & 302,809 \\
\#Records & 371,345 & 883,636 & 2,875,917 \\
Sparsity & 0.0035\% & 0.0006\% & 0.0006\% \\ \hline
\end{tabular}
\caption{Descriptive Statistics of the Amazon Dataset}
\label{amazon}
\end{table}

\subsection{User and Item Features}
Besides the records of interactions between users and items in the Imhonet dataset, the online platform also collects user information by asking them online questions. Typical question examples include preferences between two items, ideas about recent events, life styles, demographic information, etc.. From all these questions, we select eight most-popular questions and use the corresponding answers from users as user features. Meanwhile, although the online platform does not directly collect item features, we obtain these information through Google Books API\footnote{https://developers.google.com/books/}, IMDB API\footnote{http://www.omdbapi.com/} and Spotify API\footnote{https://developer.spotify.com/documentation/web-api/}, following the method in \cite{li2020ddtcdr}. To ensure the correctness of the collected item features, we validate the retrieved results with the timestamp included in the dataset. To sum up, for the Imhonet dataset, we include the user features of Age, Movie Taste, Residence, Preferred Category, Recommendation Usage, Marital Status and Personality, the book features of Category, Title, Author, Publisher, Language, Country, Price and Date, the movie features of Genre, Title, Director, Writer, Runtime, Country, Rating and Votes, and the music features of Listener, Playcount, Artist, Album, Tag, Release, Duration and Title in the recommendation process to construct user and item embeddings correspondingly.

For the Amazon dataset used in our study, we also collect the meta-information of users and items in each domain, as provided by the Amazon platform \cite{ni2019justifying}. In particular, we include the user features of rating, purchase history, review and timestamp, and the item features of category, description, title and price to generate user and item embedding respectively for the three domains.

\subsection{Baseline Models}
To conduct experiments and evaluate our model, we utilize the 5-fold cross validation and evaluate the recommendation performance based on RMSE, MAE, Precision and Recall metrics \cite{ricci2011introduction}. We compare the performance with a group of state-of-the-art methods and report the results in the next section.

\begin{itemize}
\item \textbf{DDTCDR \cite{li2020ddtcdr}} Deep Dual Transfer Cross Domain Recommendation (DDTCDR) efficiently transfers user preferences across domain pairs through dual learning mechanism.
\item \textbf{CML\cite{hsieh2017collaborative}} Collaborative Metric Learning (CML) learns a joint metric space to encode users' preferences and similarity information.
\item \textbf{CCFNet\cite{lian2017cccfnet}} Cross-domain Content-boosted Collaborative Filtering neural NETwork (CCCFNet) utilizes factorization to tie CF and content-based filtering together with a unified multi-view neural network.
\item \textbf{CDFM\cite{loni2014cross}} Cross Domain Factorization Machine (CDFM) proposes an extension of FMs that incorporates domain information in interaction patterns.
\item \textbf{CoNet\cite{hu2018conet}} Collaborative Cross Networks (CoNet) enables knowledge transfer across domains by cross connections between base networks.
\item \textbf{CMF\cite{singh2008relational}}  Collective Matrix Factorization (CMF) simultaneously factor several matrices, sharing parameters among factors when a user participates in multiple domains.
\item \textbf{NCF\cite{he2017neural}} Neural Collaborative Filtering (NCF) models latent features of users and items using collaborative filtering method and obtains satisfying recommendation performance. In this baseline, we train the NCF model separately to provide recommendations for each domain.
\item \textbf{NCF-joint\cite{he2017neural}} In this baseline, we merge the data records of both domains and jointly train the NCF model to provide recommendations for both domains simultaneously.
\end{itemize}

\subsection{Parameter Settings}
To obtain the superior recommendation performance, we conduct Bayesian Optimization \cite{NIPS2012_4522} to identify the best suitable hyperparameters for our proposed DML model as well as baseline models. Specifically, we select the dimension of feature embeddings as 16 and the batch size as 64. We construct both the encoder and the decoder as the Multi-Layer Perceptron with two hidden layers of 32 and 16 neurons respectively. The orthogonal constraint of the metric learning mapping is achieved through the Gram-Schmidt orthogonalization mechanism \cite{bjorck1994numerics,bansal2018can}.
 
\section{Results}
\subsection{Cross-Domain Recommendation Performance}
Since we have three different domains in each dataset, this results in three domain pairs for evaluation in each recommendation task. We report the experimental results of cross domain recommendation performance for DML and other baselines for the Imhonet dataset in Tables \ref{result1}, \ref{result2} and \ref{result3}, and the results for the Amazon dataset in Tables \ref{result4}, \ref{result5} and \ref{result6} respectively. As shown in these tables, the proposed DML model significantly and consistently outperforms all the baseline models across all six recommendation tasks and all the evaluation metrics. In particular in the Book-Movie recommendation scenario for the Imhonet dataset, we have observed increases of 3.01\% in the RMSE metric, 4.96\% in MAE, 2.28\% in Precision and 2.81\% in Recall over the second-best baseline approach. We have also observed similar significant improvements in other experimental settings. Therefore, we have shown that the proposed dual metric learning method works well in practice and achieves significant performance improvements vis-a-vis other state-of-the-art cross domain recommendation approaches.

\begin{table*}
\centering
\begin{tabular}{|c|cccc|cccc|} \hline
\multirow{2}{*}{Algorithm} & \multicolumn{4}{c|}{Book Dataset} & \multicolumn{4}{c|}{Movie Dataset} \\ \cline{2-9}
               &  RMSE & MAE & Pre@5 & Rec@5 & RMSE & MAE & Pre@5 & Rec@5 \\ \hline
\textbf{DML} & \textbf{0.2184*} & \textbf{0.1646*} & \textbf{0.8826*} & \textbf{0.9850*} & \textbf{0.2109*} & \textbf{0.1606*} & \textbf{0.9002*} & \textbf{0.9962*} \\ Improved \% & (+1.31\%) & (+3.63\%) & (+2.69\%) & (+3.71\%) & (+4.70\%) & (+6.30\%) & (+0.86\%) & (+0.92\%) \\ \hline
DDTCDR & 0.2213 & 0.1708 & 0.8595 & 0.9594 & 0.2213 & 0.1714 & 0.8925 & 0.9871 \\
CML & 0.2408 & 0.1927 & 0.8150 & 0.8864 & 0.2342 & 0.1924 & 0.8520 & 0.9380 \\ 
CCFNet & 0.2639 & 0.1841 & 0.8102 & 0.8872 & 0.2476 & 0.1939 & 0.8545 & 0.9300 \\ 
CDFM & 0.2494 & 0.2165 & 0.7978 & 0.8610 & 0.2289 & 0.1901 & 0.8498 & 0.9312 \\ 
CMF & 0.2921 & 0.2478 & 0.7972 & 0.8523 & 0.2738 & 0.2293 & 0.8324 & 0.9012 \\ 
CoNet & 0.2305 & 0.1892 & 0.8328 & 0.8990 & 0.2298 & 0.1903 & 0.8680 & 0.9601 \\ 
NCF & 0.2315 & 0.1887 & 0.8357 & 0.8924 & 0.2276 & 0.1895 & 0.8428 & 0.9495 \\ 
NCF-joint & 0.2378 & 0.1920 & 0.8275 & 0.8880 & 0.2254 & 0.1880 & 0.8410 & 0.9478 \\ \hline
\end{tabular}
\caption{Comparison of Cross-Domain Recommendation Performance between Book and Movie Datasets. * stands for 95\% significance.}
\label{result1}
\begin{tabular}{|c|cccc|cccc|} \hline
\multirow{2}{*}{Algorithm} & \multicolumn{4}{c|}{Book Dataset} & \multicolumn{4}{c|}{Music Dataset} \\ \cline{2-9}
               &  RMSE & MAE & Pre@5 & Rec@5 & RMSE & MAE & Pre@5 & Rec@5 \\ \hline
\textbf{DML} & \textbf{0.2162*} & \textbf{0.1615*} & \textbf{0.8649*} & \textbf{0.9645*} & \textbf{0.2689*} & \textbf{0.2255*} & \textbf{0.8475*} & \textbf{0.9070*} \\
Improved \% & (+2.13\%) & (+5.22\%) & (+0.92\%) & (+0.44\%) & (+4.94\%) & (+2.04\%) & (+0.99\%) & (+1.59\%) \\ \hline
DDTCDR & 0.2209 & 0.1704 & 0.8570 & 0.9602 & 0.2753 & 0.2302 & 0.8392 & 0.8928 \\
CML & 0.2398 & 0.1902 & 0.8112 & 0.9120 & 0.3075 & 0.2478 & 0.7922 & 0.8402 \\ 
CCFNet & 0.2630 & 0.1842 & 0.8150 & 0.9108 & 0.3090 & 0.2422 & 0.7902 & 0.8388 \\
CDFM & 0.2489 & 0.2155 & 0.8104 & 0.9102 & 0.3252 & 0.2463 & 0.7895 & 0.8365 \\
CMF & 0.2921 & 0.2478 & 0.8072 & 0.8978 & 0.3478 & 0.2698 & 0.7820 & 0.8324 \\
CoNet & 0.2307 & 0.1897 & 0.8230 & 0.9300 & 0.2801 & 0.2410 & 0.7912 & 0.8428 \\ 
NCF & 0.2315 & 0.1887 & 0.8357 & 0.8924 & 0.2828 & 0.2423 & 0.7930 & 0.8450 \\ 
NCF-joint & 0.2377 & 0.1870 & 0.8388 & 0.8890 & 0.2844 & 0.2440 & 0.7912 & 0.8428 \\ \hline
\end{tabular}
\caption{Comparison of Cross-Domain Recommendation Performance between Book and Music Datasets. * stands for 95\% significance.}
\label{result2}
\begin{tabular}{|c|cccc|cccc|} \hline
\multirow{2}{*}{Algorithm} & \multicolumn{4}{c|}{Movie Dataset} & \multicolumn{4}{c|}{Music Dataset} \\ \cline{2-9}
               &  RMSE & MAE & Pre@5 & Rec@5 & RMSE & MAE & Pre@5 & Rec@5 \\ \hline
\textbf{DML} & \textbf{0.2111*} & \textbf{0.1647*} & \textbf{0.9013*} & \textbf{0.9980*} & \textbf{0.2718*} & \textbf{0.2231*} & \textbf{0.8487*} & \textbf{0.9050*} \\
Improved \% & (+2.90\%) & (+4.24\%) & (+0.97\%) & (+1.12\%) & (+1.45\%) & (+3.46\%) & (+1.40\%) & (+1.66\%) \\ \hline
DDTCDR & 0.2174 & 0.1720 & 0.8926 & 0.9869 & 0.2758 & 0.2311 & 0.8370 & 0.8902 \\
CML & 0.2259 & 0.1870 & 0.8578 & 0.9270 & 0.3018 & 0.2590 & 0.7875 & 0.8390 \\ 
CCFNet & 0.2468 & 0.1932 & 0.8398 & 0.9310 & 0.3090 & 0.2433 & 0.7952 & 0.8498 \\
CDFM & 0.2289 & 0.1895 & 0.8306 & 0.9382 & 0.3252 & 0.2467 & 0.7880 & 0.8460 \\
CMF & 0.2738 & 0.2293 & 0.8278 & 0.9222 & 0.3478 & 0.2698  & 0.7796 & 0.8400 \\
CoNet & 0.2302 & 0.1908 & 0.8450 & 0.9508 & 0.2811 & 0.2428 & 0.8010 & 0.8512 \\ 
NCF &  0.2276 & 0.1895 & 0.8428 & 0.9495 & 0.2828 & 0.2423 & 0.7930 & 0.8450 \\
NCF-joint &  0.2238 & 0.1834 & 0.8450 & 0.9508 & 0.2877 & 0.2445 & 0.7898 & 0.8438 \\ \hline
\end{tabular}
\caption{Comparison of Cross-Domain Recommendation Performance between Movie and Music Datasets. * stands for 95\% significance.}
\label{result3}
\end{table*}

\begin{table*}
\centering
\begin{tabular}{|c|cccc|cccc|} \hline
\multirow{2}{*}{Algorithm} & \multicolumn{4}{c|}{Beauty} & \multicolumn{4}{c|}{Fashion} \\ \cline{2-9}
               & RMSE & MAE & Pre@5 & Rec@5 & RMSE & MAE & Pre@5 & Rec@5 \\ \hline
\textbf{DML} & \textbf{0.2930*} & \textbf{0.1972*} & \textbf{0.8354*} & \textbf{0.8702*} & \textbf{0.2997*} & \textbf{0.2398*} & \textbf{0.8001*} & \textbf{0.8577*} \\
Improved \% & (+1.97\%) & (+1.93\%) & (+1.49\%) & (+1.26\%) & (+0.93\%) & (+1.24\%) & (+1.30\%) & (+2.08\%) \\ \hline
DDTCDR & 0.2989 & 0.2012 & 0.8231 & 0.8594 & 0.3025 & 0.2428 & 0.7898 & 0.8402 \\
CML & 0.3033 & 0.2046 & 0.8125 & 0.8533 & 0.3077 & 0.2455 & 0.7860 & 0.8371 \\
CCCFNet & 0.3045 & 0.2035 & 0.8134 & 0.8546 & 0.3089 & 0.2460 & 0.7852 & 0.8362 \\
CDFM & 0.3022 & 0.2035 & 0.8130 & 0.8548 & 0.3102 & 0.2489 & 0.7834 & 0.8355 \\
CMF & 0.3017 & 0.2028 & 0.8122 & 0.8533 & 0.3111 & 0.2501 & 0.7798 & 0.8328 \\
CoNet & 0.2968 & 0.1982 & 0.8218 & 0.8590 & 0.3044 & 0.2433 & 0.7883 & 0.8402 \\ 
NCF & 0.3010 & 0.2023 & 0.8191 & 0.8577 & 0.3077 & 0.2472 & 0.7860 & 0.8355 \\
NCF-joint & 0.2988 & 0.2010 & 0.8205 & 0.8582 & 0.3065 & 0.2466 & 0.7877 & 0.8364 \\ \hline
\end{tabular}
\caption{Comparison of Cross-Domain Recommendation Performance between Beauty and Fashion Datasets.  `*' stands for 95\% significance.}
\label{result4}
\begin{tabular}{|c|cccc|cccc|} \hline
\multirow{2}{*}{Algorithm} & \multicolumn{4}{c|}{Beauty} & \multicolumn{4}{c|}{Art Crafts} \\ \cline{2-9}
               & RMSE & MAE & Pre@5 & Rec@5 & RMSE & MAE & Pre@5 & Rec@5 \\ \hline
\textbf{DML} & \textbf{0.2941*} & \textbf{0.1986*} & \textbf{0.8341*} & \textbf{0.8725*} & \textbf{0.2633*} & \textbf{0.2011*} & \textbf{0.8905*} & \textbf{0.9203*} \\
Improved \% & (+1.70\%) & (+1.39\%) & (+1.48\%) & (+1.60\%) & (+1.98\%) & (+3.87\%) & (+2.16\%) & (+1.19\%) \\ \hline
DDTCDR & 0.2992 & 0.2014 & 0.8219 & 0.8588 & 0.2680 & 0.2092 & 0.8717 & 0.9095 \\
CML & 0.3065 & 0.2044 & 0.8122 & 0.8533 & 0.2729 & 0.2130 & 0.8644 & 0.9011 \\
CCCFNet & 0.3044 & 0.2031 & 0.8137 & 0.8540 & 0.2717 & 0.2119 & 0.8655 & 0.9039 \\
CDFM & 0.3021 & 0.2037 & 0.8130 & 0.8548 & 0.2708 & 0.2127 & 0.8652 & 0.9042 \\
CMF & 0.3021 & 0.2029 & 0.8127 & 0.8532 & 0.2726 & 0.2133 & 0.8637 & 0.9013 \\
CoNet & 0.2977 & 0.1982 & 0.8211 & 0.8592 & 0.2689 & 0.2080 & 0.8723 & 0.9049 \\ 
NCF & 0.3017 & 0.2029 & 0.8197 & 0.8570 & 0.2707 & 0.2111 & 0.8699 & 0.9071 \\
NCF-joint & 0.3001 & 0.2009 & 0.8203 & 0.8588 & 0.2695 & 0.2099 & 0.8711 & 0.9092 \\ \hline
\end{tabular}
\caption{Comparison of Cross-Domain Recommendation Performance between Beauty and Art Crafts Datasets.  `*' stands for 95\% significance.}
\label{result5}
\begin{tabular}{|c|cccc|cccc|} \hline
\multirow{2}{*}{Algorithm} & \multicolumn{4}{c|}{Fashion} & \multicolumn{4}{c|}{Art Crafts} \\ \cline{2-9}
               & RMSE & MAE & Pre@5 & Rec@5 & RMSE & MAE & Pre@5 & Rec@5 \\ \hline
\textbf{DML} & \textbf{0.2984*} & \textbf{0.2375*} & \textbf{0.8023*} & \textbf{0.8514*} & \textbf{0.2622*} & \textbf{0.2030*} & \textbf{0.8895*} & \textbf{0.9233*} \\
Improved \% & (+1.29\%) & (+2.22\%) & (+1.56\%) & (+1.37\%) & (+2.46\%) & (+2.96\%) & (+2.09\%) & (+1.60\%) \\ \hline
DDTCDR & 0.3023 & 0.2429 & 0.7900 & 0.8399 & 0.2688 & 0.2092 & 0.8713 & 0.9088 \\
CML & 0.3092 & 0.2477 & 0.7839 & 0.8337 & 0.2725 & 0.2136 & 0.8637 & 0.9010 \\
CCCFNet & 0.3087 & 0.2462 & 0.7849 & 0.8357 & 0.2711 & 0.2116 & 0.8664 & 0.9046 \\
CDFM & 0.3100 & 0.2488 & 0.7831 & 0.8359 & 0.2709 & 0.2129 & 0.8651 & 0.9040 \\
CMF & 0.3112 & 0.2506 & 0.7801 & 0.8333 & 0.2731 & 0.2137 & 0.8637 & 0.9015 \\
CoNet & 0.3048 & 0.2439 & 0.7886 & 0.8394 & 0.2683 & 0.2071 & 0.8726 & 0.9048 \\ 
NCF & 0.3079 & 0.2477 & 0.7859 & 0.8357 & 0.2707 & 0.2120 & 0.8693 & 0.9062 \\ 
NCF-joint & 0.3061 & 0.2449 & 0.7870 & 0.8381 & 0.2691 & 0.2103 & 0.8705 & 0.9074 \\ \hline
\end{tabular}
\caption{Comparison of Cross-Domain Recommendation Performance between Fashion and Art Crafts Datasets.  `*' stands for 95\% significance.}
\label{result6}
\end{table*}

\subsection{Impact of the Number of Overlap Users}
As we have discussed in this paper, one important advantage of the proposed model lies in its ability to significantly reduce the requirement of overlap user information for the cross-domain recommendation tasks. To empirically demonstrate this point, we conduct additional experiments where we include different numbers of overlap users into the recommendation process, and discard the records from other unselected overlap users from the dataset. These included overlap users are randomly selected from the pool of overlap users. Our results in Figure \ref{num} show that the proposed DML model indeed works well with very few overlap users. When we do not include any overlap users into the learning process, the cross-domain recommendation performance is not particularly satisfying. When we start to include a few overlap users into the cross-domain recommendation model, we would observe rapid and significant improvements of recommendation performance. After we add more overlap users into the model, we would witness the phenomenon of ''elbow point'', where additional overlap users only incrementally improve the recommendation performance. In particular, when we only include 8 overlap users and use DML model to provide cross-domain recommendations, we still achieve satisfying recommendation performance and outperform all the baseline models that utilize all the overlap users, though the improvements are not as much as those when we also include all overlap information in the DML model. Eventually when we utilize all the overlap users in the datasets, we achieve the state-of-the-art recommendation performance and outperform all other baselines consistently and significantly. This observation validates our claim and significantly enhances the practicality of our proposed model.

\begin{figure}[h]
    \centering
    \begin{subfigure}[t]{0.23\textwidth}
        \centering
        \includegraphics[width=\textwidth]{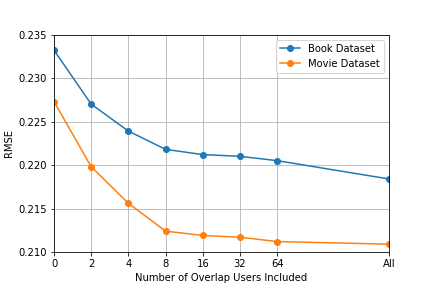}
        \caption{Book-Movie Domains}
    \end{subfigure}%
    ~
    \begin{subfigure}[t]{0.23\textwidth}
        \centering
        \includegraphics[width=\textwidth]{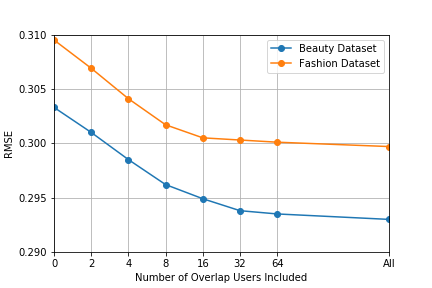}
        \caption{Beauty-Fashion Domains}
    \end{subfigure}%
\caption{Cross-Domain Recommendation Performance versus Number of Overlap Users Included}
\label{num}
\end{figure}

\subsection{Impact of Including/Excluding Explicit Feature Information}
In our experiments, we collect and include explicit feature information of users and items into the recommendation process, and we manage to achieve satisfying recommendation performance. However, in some recommendation scenarios, explicit features might not be available, which could potentially jeopardize the recommendation performance. To test for its impact, we conduct an ablation study in this paper in which we exclude all the explicit features in the learning process, and the latent embeddings are generated directly from the user and item IDs correspondingly. We also exclude explicit features from our selected baselines during their training process for the fair comparison. As we show in Table \ref{feature}, in the case when no feature information is used, our proposed DML model still significantly and consistently outperforms the baseline approaches in terms of all four evaluation metrics, illustrating the superiority and the robustness of our proposed approach. In addition, we observe that the recommendation performance is significantly worse than that in Table \ref{result1}, demonstrating the importance of including explicit feature information for solving cross domain recommendation tasks.

\begin{table*}
\centering
\begin{tabular}{|c|cccc|cccc|} \hline
\multirow{2}{*}{Algorithm} & \multicolumn{4}{c|}{Book Dataset} & \multicolumn{4}{c|}{Movie Dataset} \\ \cline{2-9}
               &  RMSE & MAE & Pre@5 & Rec@5 & RMSE & MAE & Pre@5 & Rec@5 \\ \hline
\textbf{DML} & \textbf{0.2395*} & \textbf{0.1801*} & \textbf{0.8644*} & \textbf{0.9535*} & \textbf{0.2382*} & \textbf{0.1888*} & \textbf{0.8714*} & \textbf{0.9622*} \\
Improved \% & (+6.73\%) & (+9.81\%) & (+1.97\%) & (+3.21\%) & (+5.10\%) & (+5.46\%) & (+1.46\%) & (+3.27\%) \\ \hline
DDTCDR & 0.2568 & 0.1997 & 0.8477 & 0.9238 & 0.2510 & 0.1997 & 0.8589 & 0.9317 \\
CML & 0.2712 & 0.2114 & 0.8011 & 0.8710 & 0.2689 & 0.2115 & 0.8322 & 0.9008 \\ 
CDFM & 0.2791 & 0.2326 & 0.7739 & 0.8499 & 0.2571 & 0.2169 & 0.8180 & 0.8999 \\ 
CMF & 0.3112 & 0.2701 & 0.7739 & 0.8345 & 0.2984 & 0.2500 & 0.8009 & 0.8733 \\ 
CoNet & 0.2611 & 0.2099 & 0.8150 & 0.8866 & 0.2498 & 0.2124 & 0.8374 & 0.9298 \\ 
NCF & 0.2635 & 0.2135 & 0.8293 & 0.8798 & 0.2588 & 0.2121 & 0.8222 & 0.9079 \\ \hline
\end{tabular}
\caption{Comparison of Cross-Domain Recommendation Performance Excluding Feature Information. * stands for 95\% significance.}
\label{feature}
\end{table*}

\subsection{Impact of Different Recommendation Networks for Recommendations}
To validate that the improvement gained in our model is not sensitive to the specific selection of neural recommendation models, we also conduct additional experiments to examine the change of recommendation performance corresponding to different choices of $RS_A$ and $RS_B$, including MLP that we use in this paper, CNN \cite{ying2018graph}, RNN \cite{hidasi2015session}, Neural Attention \cite{zhu2019dan} and Adversarial Network \cite{wang2018neural} to learn user preferences and produce final recommendations. As shown in Table \ref{network}, the cross-domain recommendation performance between book and movie domains is generally consistent across all the different types of autoencoders, albeit some small numerical differences in recommendation performance take place where the neural attention model performs slightly better than other models - due to its capability to capture dynamic user preferences. Nevertheless, the proposed DML model works well in all these experimental settings, and significantly outperforms all the considered cross-domain recommendation baselines.

\begin{table*}
\centering
\begin{tabular}{|c|cccc|cccc|} \hline
\multirow{2}{*}{Recommendation Network} & \multicolumn{4}{c|}{Book Dataset} & \multicolumn{4}{c|}{Movie Dataset} \\ \cline{2-9}
               &  RMSE & MAE & Pre@5 & Rec@5 & RMSE & MAE & Pre@5 & Rec@5 \\ \hline
Multi-Layer Perceptron &  0.2184 & 0.1646 & 0.8826 & 0.9850 & 0.2109 & 0.1606 & 0.9002 & 0.9962 \\ \hline
Convolutional Neural Network &  0.2186 &  0.1648 &  0.8823 & 0.9848 & 0.2107 & 0.1609 & 0.8998 &  0.9962\\ \hline
Recurrent Neural Network & 0.2184  & 0.1648  & 0.8830  & 0.9850 & 0.2111 & 0.1606 & 0.9004 & 0.9952 \\ \hline
Neural Attention & 0.2165 & 0.1632  & 0.8847 & 0.9858 & 0.2097 & 0.1587 & 0.9020 & 0.9980 \\ \hline
Adversarial Network & 0.2192 & 0.1653 & 0.8815 & 0.9848 & 0.2119 & 0.1609 & 0.8998 & 0.9952 \\ \hline
\end{tabular}
\newline
\caption{Comparison of Different Recommendation Networks for Cross-Domain Recommendations in the Imhonet dataset}
\label{network}
\end{table*}

\subsection{Impact of Different Types of AutoEncoder for Feature Embeddings}
To validate that the improvement gained in our model is not sensitive to the specific selection of the autoencoder network, we also conduct additional experiments to examine the change of recommendation performance corresponding to different types of autoencoders, including AE \cite{sutskever2014sequence}, VAE \cite{kingma2013auto}, AAE \cite{makhzani2015adversarial}, WAE \cite{tolstikhin2017wasserstein} and HVAE \cite{davidson2018hyperspherical} to construct the feature embeddings for users and items. As shown in Table \ref{autoencoder}, the cross-domain recommendation performance remains consistent across all the different selection of autoencoder models, where we do not witness any significant differences in any evaluation metrics. This observation also holds true other recommendation tasks in our experiments. Therefore, the choice of particular autoencoder method is not relevant to the cross-domain recommendation performance in our proposed model.

\begin{table*}
\centering
\begin{tabular}{|c|cccc|cccc|} \hline
\multirow{2}{*}{Types of AutoEncoder} & \multicolumn{4}{c|}{Book Dataset} & \multicolumn{4}{c|}{Movie Dataset} \\ \cline{2-9}
               &  RMSE & MAE & Pre@5 & Rec@5 & RMSE & MAE & Pre@5 & Rec@5 \\ \hline
AutoEncoder &  0.2184 & 0.1646 & 0.8826 & 0.9850 & 0.2109 & 0.1606 & 0.9002 & 0.9962 \\ \hline
Variational AutoEncoder &  0.2190 &  0.1654 &  0.8796 & 0.9848 & 0.2107 & 0.1609 & 0.8998 &  0.9948\\ \hline
Adversarial AutoEncoder & 0.2186  & 0.1641  & 0.8805  & 0.9836 & 0.2115 & 0.1599 & 0.9008 & 0.9952 \\ \hline
Wasserstein AutoEncoder & 0.2186 & 0.1649  & 0.8822 & 0.9832 & 0.2119 & 0.1597 & 0.9002 & 0.9952 \\ \hline
Hyperspherical Variational AutoEncoder & 0.2190 & 0.1641 & 0.8816 & 0.9848 & 0.2104 & 0.1606 & 0.8998 & 0.9960 \\ \hline
\end{tabular}
\newline
\caption{Comparison of Autoencoder Settings for Cross-Domain Recommendations in the Imhonet dataset}
\label{autoencoder}
\end{table*}

\subsection{Sensitivity Analysis}
As we theoretically discussed in Section 4.1, the number of overlap users required by the proposed model for producing cross-domain recommendations would be determined by the dimension of latent embeddings. If the size of the embedding is very large, we might still need to obtain a large number of overlap users. Therefore, we conduct additional analysis in this section to explore  sensitivity of the dimension of latent embeddings towards cross domain recommendation performance. In particular, we present the recommendation performance for the Imhonet dataset in Figure \ref{dimension} using the dimension sizes of 4, 8, 16, 32, 64, 128 and 256 respectively to learn the user and item embeddings. We observe that increasing the size of the latent dimensions would indeed improve the recommendation performance in general. However, these improvements are only marginal in their absolute values, and we can achieve reasonable and satisfying cross-domain recommendation using 16 or 32 dimensions of latent vectors. Therefore, our theoretical analysis could be applied to our cross-domain recommendation tasks, and ablation studies in the previous section further confirm that we could provide satisfying recommendation performance using as little as 8 overlap users. To sum up, our proposed model would not be sensitive to the specific selection of the dimension size.

\begin{figure}[h]
    \centering
    \begin{subfigure}[t]{0.23\textwidth}
        \centering
        \includegraphics[width=\textwidth]{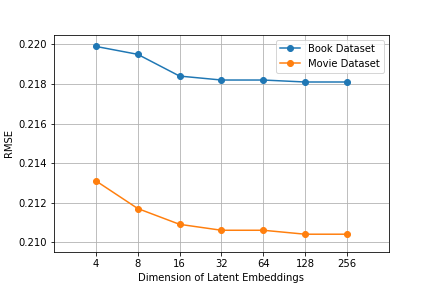}
        \caption{Book-Movie Domains}
    \end{subfigure}%
    ~
    \begin{subfigure}[t]{0.23\textwidth}
        \centering
        \includegraphics[width=\textwidth]{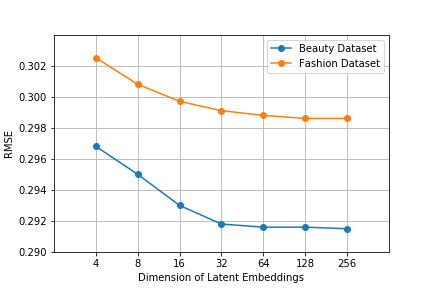}
        \caption{Beauty-Fashion Domains}
    \end{subfigure}%
\caption{Sensitivity Analysis of Different Sizes of Latent Embeddings of the DML model.}
\label{dimension}
\end{figure}

\subsection{Scalability Analysis}
To test for scalability, we provide cross-domain recommendations using DML with aforementioned hyperparameter values for the Imhonet dataset of Book, Movie and Music domains with increasing data sizes from 100 to 1,000,000 records. As shown in Figure \ref{scalability}, we empirically observe that the proposed DML model scales linearly with the increase in the number of data records for the model training process. The training procedure comprises of generating feature embeddings for users and items, as well as the dual metric learning stage jointly optimized across different domains, as we describe in Section 3. The optimization phase is made efficient using batch normalization \cite{ioffe2015batch} and distributed training. As our experiments confirm, DML is capable of providing cross-domain recommendations efficiently and indeed scales well.

\begin{figure}[h]
    \centering
    \begin{subfigure}[t]{0.23\textwidth}
        \centering
        \includegraphics[width=\textwidth]{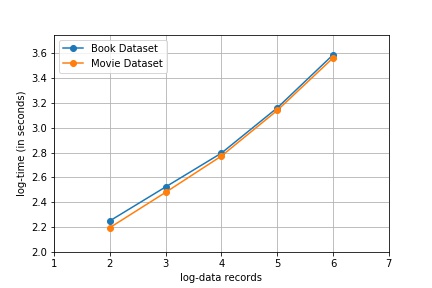}
        \caption{Book-Movie Domains}
    \end{subfigure}%
    ~
    \begin{subfigure}[t]{0.23\textwidth}
        \centering
        \includegraphics[width=\textwidth]{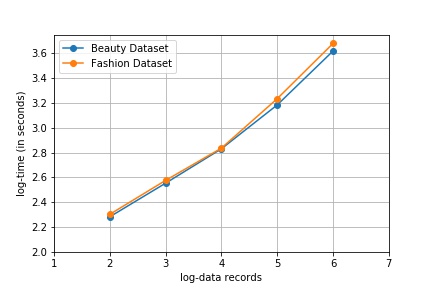}
        \caption{Beauty-Fashion Domains}
    \end{subfigure}%
\caption{Scalability Analysis of the DML model with increasing data sizes from 100 to 1,000,000 records.}
\label{scalability}
\end{figure}

\subsection{Convergence Analysis}
While we have provided the technical foundations of the convergence proposition of the dual learning method under the classical setting in Section 4.2, this proposition might not be directly applicable to our proposed model, as we utilize deep learning techniques for which the optimization process is different. Therefore, we test the convergence empirically in our study, and conjecture that it should still happen even though our method does not explicitly satisfy the condition. In particular, we train the model iteratively for 100 epochs until the change of loss function is less than 1e-5. We plot the training loss over time for our cross-domain recommendation tasks in Figure \ref{epoch}. The key observation is that, DML starts with relatively higher loss due to the noisy initialization. As times goes by, DML manages to stabilize quickly and significantly achieves satisfying recommendation performance only after 10 epochs. Therefore, we empirically validate that the proposed DML indeed converges in our experimental settings.

\begin{figure}[h]
    \centering
    \begin{subfigure}[t]{0.23\textwidth}
        \centering
        \includegraphics[width=\textwidth]{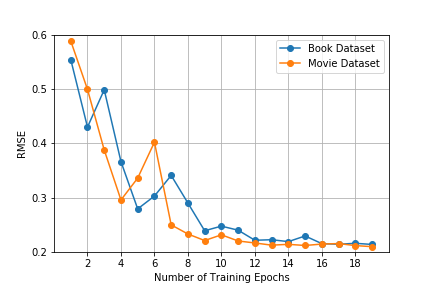}
        \caption{Book-Movie Domains}
    \end{subfigure}%
    ~
    \begin{subfigure}[t]{0.23\textwidth}
        \centering
        \includegraphics[width=\textwidth]{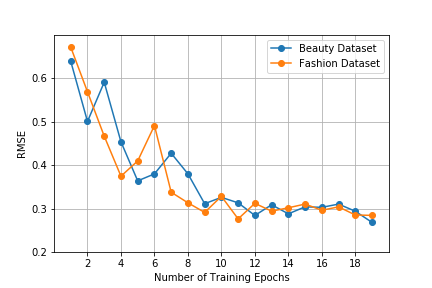}
        \caption{Beauty-Fashion Domains}
    \end{subfigure}%
\caption{Convergence Analysis of the DML model with Different Training Epochs.}
\label{epoch}
\end{figure}

\section{Conclusion}
In this paper, we propose a novel dual metric learning based model DML that significantly improves recommendation performance across different domains and reduces the requirement of overlapping user information. We accomplish this by proposing a novel latent orthogonal metric mapping and iteratively going through the learning loop until the models stabilize for both domains. We also prove that this convergence is guaranteed under certain conditions and empirically validate the hypothesis for our model across different experimental settings. We show that the proposed model could be extended for multi-domain tasks.

Note that the proposed approach provides several benefits, including that it (a) transfers information about latent interactions instead of explicit features from the source domain to the target domain; (b) utilizes the dual learning mechanism to enable the bidirectional training that improves performance measures for both domains simultaneously; (c) learns the latent orthogonal mapping function across two domains, that (i) preserves similarity of user preferences and thus enables proper transfer learning and (ii) computes the inverse mapping function efficiently; (d) works well with very few overlap users (as low as 8 in our experiments) and obtains the results comparable to the state-of-the-art baselines. This makes our approach more practical vis-a-vis the previously proposed cross-domain recommendation models because many business applications have very few overlapping users.

As the future work, we plan to further study the application of metric learning to provide even better cross-domain recommendations. In addition, we plan to extend the convergence proposition to more general settings.

\bibliographystyle{IEEEtran}
\bibliography{sigproc}

\begin{IEEEbiography}[{\includegraphics[width=1in,height=1.25in,clip,keepaspectratio]{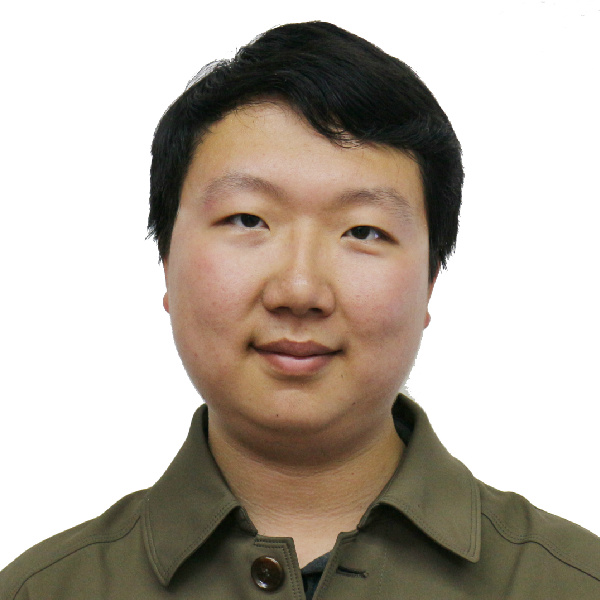}}]{Pan Li}
Pan Li is a PhD student in the department of Technology, Operation and Statistics, Stern School of Business, New York University. His research interests include recommender systems and business analytics.  He has published papers in international conferences and journals including WSDM, EMNLP, RecSys, AAAI, IEEE TKDE, ACM TIST and ACM TMIS.
\end{IEEEbiography}

\begin{IEEEbiography}[{\includegraphics[width=1in,height=1.25in,clip,keepaspectratio]{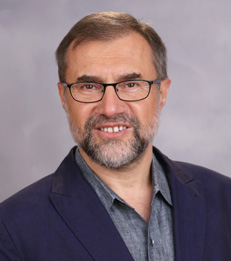}}]{Alexander Tuzhilin}
Alexander Tuzhilin is the Leonard N. Stern Professor of Business in the Department of Technology, Operations and Statistics at the Stern School of Business, NYU. His research interests include personalization, recommender systems, machine learning and AI, where he has produced over 150 research publications. He has served on the organizing committees of numerous conferences, including as the Program and the General Chair of ICDM, and as the Program and the Conference Chair of RecSys. He has also served on the editorial boards of several journals, including as the Editor-in-Chief of the ACM TMIS.
\end{IEEEbiography}

\end{document}